\documentclass[runningheads]{llncs}

\usepackage[T1]{fontenc}
\usepackage[utf8]{inputenc}
\usepackage{graphicx}
\usepackage{color}

\usepackage{amsmath, amssymb} 
\usepackage{mathtools}

\usepackage{tikz}
\usetikzlibrary{decorations.pathreplacing}
\usetikzlibrary{positioning, arrows.meta}
\usetikzlibrary{calc}

\usepackage{hyperref}

\usepackage{cleveref} %
\newcommand{\NP}{\ensuremath{\mathsf{NP}}}

\usepackage{booktabs}
\usepackage{multirow}
\usepackage{xspace}
\usepackage{nicefrac}

\newcommand{\GED}{\ensuremath{\mathrm{GED}}\xspace}

\newcommand{\MCIS}{\ensuremath{\mathrm{MCIS}}\xspace}
\newcommand{\SI}{\ensuremath{\mathrm{SI}}\xspace}

\newcommand{\MCES}{\ensuremath{\mathrm{MCES}}\xspace}

\newcommand{\GSim}{\textsc{GSim}\xspace}

\begin{document}

\title{On the Complexity of Graph Edit Distance in Restricted Graph Classes}
\author{Maximilian Limmer\inst{1} \and
Nils~M.~Kriege\inst{1,2}\orcidID{0000-0003-2645-947X}}
\authorrunning{M. Limmer, N.M. Kriege}
\institute{Faculty of Computer Science, University of Vienna, Vienna, Austria \and
Research Network Data Science, University of Vienna, Vienna, Austria
\email{a12531403@unet.univie.ac.at}, \email{nils.kriege@univie.ac.at}}
\maketitle              %

\begin{abstract}
The graph edit distance generalizes several well-known NP-hard problems and is therefore NP-hard itself. 
However, the relationship between the considered graph class, the edit cost function, and the resulting computational complexity is not well understood.
We investigate this interplay by revisiting polynomial-time reductions from the literature, which reduce subgraph isomorphism and maximum common induced subgraph to the graph edit distance.
For these classical problems, a sharp distinction between NP-hard and polynomial-time solvable cases is known, and we make the implications for the complexity of the graph edit distance explicit.
We establish a graph-class-preserving correspondence between the maximum common edge subgraph and graph edit distance under a specific cost function, both in labeled and unlabeled graphs. In the unlabeled setting, the maximum common edge subgraph problem is polynomial-time solvable when one graph is a path and the other is a tree.
In contrast, for labeled graphs, we prove that both the maximum common edge subgraph and the graph edit distance remain NP-hard, even when both graphs are paths.

\keywords{Graph Edit Distance \and Maximum Common Subgraph \and Graph Similarity \and Computational Complexity.}
\end{abstract}

\section{Introduction}
The graph edit distance (\GED) is one of the most general measures of dissimilarity between graphs and is widely used in practice. It quantifies the dissimilarity between two graphs as the minimum cost of a sequence of edit operations turning one graph into the other. The edit operations---insertion, deletion, and relabeling of vertices and edges---are assigned application-dependent costs.
It is known that both the classical subgraph isomorphism problem (\SI) and the maximum common induced subgraph problem (\MCIS) can be reduced to computing the \GED under suitable cost functions~\cite{bunkeRelationGraphEdit1997,zengComparingStarsApproximating2009}. Since both problems are well-known to be \NP-hard~\cite{garey1979computers}, computing \GED is also \NP-hard.
The complexity landscapes of \SI and \MCIS are well understood for many restricted graph classes and parameter combinations such as treewidth and degree~\cite{garey1979computers,SYSLO198291,MATOUSEK1992343,guptaComplexitySubgraphIsomorphism1996}.

It is not clear to what extent these results transfer to \GED using the existing reductions, nor how the interaction between graph classes and specific cost functions affects the computational complexity.
In particular, the choice of edit costs can fundamentally alter the nature of the problem. 
As a simple example, consider a cost function $C_0$, see Table~\ref{tab:cost-functions}, under which edge insertions, deletions, and relabelings incur zero cost. In this case, computing the GED reduces to finding an optimal assignment between two vertex sets and can be solved in cubic time using the Hungarian algorithm. Similarly, \GED can be solved efficiently on edgeless graphs for arbitrary cost functions.
In practice, unit cost functions are widely used, and various GED-based similarity search methods are restricted to such functions, e.g., \cite{zengComparingStarsApproximating2009,LiangZ17}, with only a few exceptions~\cite{BauseBSK21}. For the edge cost function, we show an equivalence to the \MCES problem, which is fundamental for the comparison of molecular graphs~\cite{labeledLineGraph,kriege_chemical_2025}.

In this work, we investigate the complexity of computing the \GED under non-trivial cost functions in restricted graph classes. We first revisit existing reductions from \SI and \MCIS and derive their implications for the complexity of \GED. While \MCIS generalizes \SI and remains \NP-hard even on trees, these hardness results do not carry over to \GED via the known reduction by Bunke~\cite{bunkeRelationGraphEdit1997}. Therefore, we establish a reduction from the maximum common edge subgraph problem (\MCES) that preserves the graph classes of the input graphs. This allows us to obtain the result that computing \GED is \NP-hard even when both input graphs are trees, in both the labeled and unlabeled settings. For unlabeled graphs, or equivalently for a cost function permitting relabelings at zero cost, we show that a known polynomial-time \MCES algorithm yields a polynomial-time algorithm for \GED for a specific cost function when one graph is a path, and the other is a tree. Finally, for labeled graphs, we prove that the \GED with non-zero relabeling costs remains \NP-hard even when both input graphs are paths.

\begin{figure}[t]
\centering
\begin{tikzpicture}[
  every node/.style = {
    draw,
    rounded corners = 4pt,
    align = center,
    font = \small,
    inner sep = 2pt,
    minimum width = 2cm,
    minimum height = .7cm,
  },
  arr/.style = {
    ->,
    >= {Stealth[length=7pt]}
  },
  doublearr/.style = {
    <->,
    >= {Stealth[length=7pt]}
  },
]

\node                                       (GED)  {\GED}; %
\node[above left  = 0.5cm and 2.8cm of GED] (SI)   {\SI}; %
\node[below left  = 0.5cm and 2.8cm of GED] (GSim) {\GSim}; %
\node[above right = 0.5cm and 2.8cm of GED] (MCIS) {\MCIS}; %
\node[below right = 0.5cm and 2.8cm of GED] (MCES) {\MCES}; %

\draw[arr, ultra thick]
  (SI) -- node[above, font=\scriptsize, draw=none, fill=none,
               inner sep=2pt] {Eq.~\eqref{eq:SItoGED}~\cite{zengComparingStarsApproximating2009}}
  (GED);

\draw[arr, dashed]
  ([xshift=-.5cm]MCES.north) -- node[right, font=\scriptsize, draw=none, fill=none, xshift=-.5cm] {Eq.~\eqref{eq:MCES_MCIS} \\[1mm] \cite{Nicholsonetal1987}}
  ([xshift=-.5cm]MCIS.south);
\draw[arr]
  ([xshift=.5cm]MCES.north) -- node[right, font=\scriptsize, draw=none, fill=none, xshift=-.45cm] {Eq.~\eqref{eq:MCES_MCIS} \\[1mm] \cite{Nicholsonetal1987,labeledLineGraph}}
  ([xshift=.5cm]MCIS.south);

\draw[doublearr, ultra thick, dashed]
  (GSim) -- node[above, font=\scriptsize, draw=none, fill=none,
                 inner sep=2pt] {Eq.~\eqref{eq:GSimToMCES}}
  (MCES);

\draw[arr]
  (MCIS) -- node[above, font=\scriptsize, draw=none, fill=none,
                inner sep=2pt] {Eq.~\eqref{eq:bunke}~\cite{bunkeRelationGraphEdit1997}}
  (GED);

\draw[arr, ultra thick]
  (MCES) -- node[above, font=\scriptsize, draw=none, fill=none,
                inner sep=2pt] {Thm.~\ref{thm:GEDMCES}}
  (GED);

\end{tikzpicture}
\caption{Polynomial-time reductions between graph similarity problems. 
Bold arrows preserve graph classes; dashed arrows apply to unlabeled graphs only.}
\label{fig:reductions}
\end{figure}
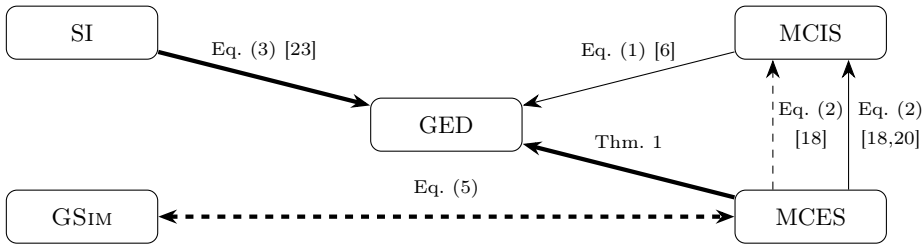

\section{Background and Problem Relationships}
\label{sec:background}

We formally define the graph edit distance and related isomorphism problems, and summarize known correspondences between them.

\subsection{Preliminaries}
\label{sec:prelims}
We consider simple undirected graphs with categorical vertex and edge labels.
A graph $G$ is a tuple $(V,E)$, where $V$ is the set of vertices (or nodes) and $E$ is the set of edges. An edge $(u,v)$ connects two vertices $u,v$ in $V$, where $(u,v)$ and $(v,u)$ refer to the same edge. We endow a graph with a labeling function $\ell\colon V \cup E \to \Sigma$, where $\Sigma$ is a finite alphabet, and write $\ell(v)$ for the label of a vertex $v$ in $V$ and $\ell(u,v)$ for the label of an edge $(u,v)$ in $E$.

An \emph{edit operation} changes a graph via insertion, deletion, or relabeling of 
vertices and edges, where each operation incurs non-negative cost as assigned by the cost function $C$.
An \emph{edit path} $\gamma = (o_1,\ldots,o_k)$ from a graph $G_1$ to a graph $G_2$ is a sequence of
edit operations transforming $G_1$ into a graph isomorphic to $G_2$, with total cost $C(\gamma) = \sum_i C(o_i)$.  
The \emph{graph edit distance} (\GED) is defined as
\[
  \mathrm{GED}(G_1, G_2; C)
  = \min_{\gamma \in \Gamma(G_1, G_2)} C(\gamma),
\]
where $\Gamma(G_1, G_2)$ is the set of all restricted edit paths from $G_1$ to $G_2$. Restricted edit paths abide by seven structural constraints and guarantee the equivalence to an algorithmically amenable node map formulation of the GED~\cite{bougleuxGraphEditDistance2017,BlumenthalG20}. In a restricted edit path, a vertex can only be deleted after its incident edges have been deleted (C1), an edge can only be inserted if its endpoints exist (C2), each node and edge is edited at most once (C3--6), and deletions and reinsertions of edges between relabeled nodes are forbidden (C7), see~\cite{bougleuxGraphEditDistance2017} for details on these constraints.
We consider cost functions that assign a constant relabeling cost for nodes with non-matching labels and constant insertion and deletion costs. Specifically, we discuss the cost functions defined in Table~\ref{tab:cost-functions} in the context of reductions between several related problems defined in the following.

Two graphs $G_1=(V_1, E_1)$ and $G_2=(V_2,E_2)$ with labeling functions $\ell_1$ and $\ell_2$ are \emph{isomorphic}, written $G_1 \simeq G_2$, if there is a bijection $\phi\colon V_1 \to V_2$ such that $\ell_1(v) = \ell_2(\phi(v))$ for all $v$ in $V_1$, and $(u,v)$ in $E_1$ if and only if $(\phi(u), \phi(v))$ in $E_2$ for all $u$, $v$ in $V_1$, and $\ell_1(u,v)=\ell_2(\phi(u),\phi(v))$ for all $(u,v)$ in $E_1$.
Given two graphs $G_1$ and $G_2$, the \emph{subgraph isomorphism problem} (\SI) is to decide if $G_2$ contains a subgraph isomorphic to $G_1$. A \emph{common (induced) subgraph} of $G_1$ and $G_2$ is an (induced) subgraph of $G_1$ that is isomorphic to an (induced) subgraph of $G_2$.
The \emph{maximum common induced subgraph problem} (\MCIS) is to find the maximum number of vertices in a common induced subgraph of $G_1$ and $G_2$ and we denote the result by $\mathrm{MCIS}(G_1, G_2)$.
The \emph{maximum common edge subgraph problem} (\MCES) asks for the maximum number of edges in a common subgraph of $G_1$ and $G_2$, and we write $\mathrm{MCES}(G_1, G_2)$.  Note that a common subgraph may contain isolated vertices, which, however, do not contribute to the solution size. Hence, there exists an edge-induced subgraph without isolated vertices with the same number of edges. 
Finally, the \emph{graph similarity problem} ($\GSim$), as considered, e.g., in \cite{grohe_et_al:LIPIcs.MFCS.2018.20}, minimizes the Frobenius distance between the adjacency matrices $A_{G_1}$ and $A_{G_2}$ of two graphs $G_1$ and $G_2$ of order $n$, i.e.,
\[
  \mathrm{dist}(G_1, G_2) := \min_{\pi \in S_n}
  \|\,A^{\pi}_{G_1} - A_{G_2}\|_F\,,
\]
where $A^{\pi}_{G_1}$ denotes the matrix obtained from $A_{G_1}$ by permuting both rows and columns according to the permutation $\pi$, and $S_n$ denotes the symmetric group on $n$ objects.

\subsection{Known Reductions and Implications on Complexity}
\label{sec:reductions}
Several reductions of classical combinatorial isomorphism problems to \GED have been proposed, partly for their intrinsic value and partly to derive complexity results. We summarize these results and their direct implications for the complexity of \GED. 
Table~\ref{tab:cost-functions} summarizes the cost functions underlying the reductions, and Figure~\ref{fig:reductions} provides a complete overview of the results, including those derived later in this work.

\begin{table}[t]
\centering \setlength{\tabcolsep}{3pt}
\caption{Overview of constant edit cost functions.}
\label{tab:cost-functions}

\begin{tabular}{llcccccc}
\toprule
\multicolumn{2}{c}{Cost function} & \multicolumn{3}{c}{Vertex} & \multicolumn{3}{c}{Edge} \\
\cmidrule(lr){1-2}\cmidrule(lr){3-5}\cmidrule(lr){6-8}
Name & Symbol & Insertion & Deletion & Relabeling & Insertion & Deletion & Relabeling \\
\midrule
Bunke~\cite{bunkeRelationGraphEdit1997}& $C_B$ & 1 & 1 & $\infty$ & 0 & 0 & $\infty$ \\
Unit & $C_U$  & 1 & 1 & 1 & 1 & 1 & 1 \\
Edge & $C_E$  & 0 & 0 & $\infty$ & 1 & 1 & 2 \\
Zero & $C_0$  & 1 & 1 & 1 & 0 & 0 & 0 \\
\bottomrule
\end{tabular}
\end{table}

\subsubsection{\MCIS reduces to \GED.}
Bunke~\cite{bunkeRelationGraphEdit1997} has shown that MCIS reduces to GED as follows.
Let $\widetilde{G}$ denote the complete graph obtained from $G$ by adding edges with a 
specific label $\bot$ between all pairs of non-adjacent nodes.
Then, under the Bunke cost function $C_B$, it holds
\begin{equation}
  \label{eq:bunke}
  \mathrm{GED}(\widetilde{G}_1, \widetilde{G}_2; C_B) = |V_1| + |V_2| - 2\,\MCIS(G_1, G_2).
\end{equation}
Note that the assumption of restricted edit paths in combination with the complete graph model is crucial to guarantee that all adjacencies and non-adjacencies of the induced subgraph are 
preserved.
Since MCIS is \NP-hard, we may conclude that \GED is \NP-hard. However, the complexity results for MCIS in restricted graph classes cannot be transferred from this result, since the construction uses complete input graphs.

\subsubsection{\MCES reduces via \MCIS to \GED.}
\MCES can be reduced to \MCIS via line graphs using a transformation proposed in~\cite{Nicholsonetal1987}, which 
is widely-used in cheminformatics~\cite{kriege_chemical_2025}.
The \emph{line graph} $L(G)$ of a graph $G=(V, E)$ is the graph with vertex 
set $E(G)$, in which two vertices are adjacent if and only if the two corresponding edges are adjacent in $G$.
Two isomorphic graphs clearly have isomorphic line graphs. 
A classical result by Whitney~\cite{whitney_congruent_1932} states that for two 
connected graphs $G$ and $H$, the reverse holds in almost all cases: 
If $L(G) \simeq L(H)$ then $G \simeq H$, or $G \simeq K_{1,3}$ and 
$H \simeq K_3$, or vice versa. 
Each induced subgraph of $L(G)$ directly corresponds to an edge-induced subgraph of $G$ and vice versa. 
Hence, we have
\begin{equation}
  \label{eq:MCES_MCIS}
  \mathrm{MCES}(G_1, G_2) = \MCIS(L(G_1), L(G_2)),
\end{equation}
unless all the corresponding isomorphisms map a $K_{1,3}$ to $K_3$, or vice versa. 
Note that this exception is relevant only for components of the \MCES consisting of three edges and can be easily checked.
The transformation can be extended to cope with labeled graphs by creating labeled
line graphs~\cite{labeledLineGraph}.
Note that the transformation does not guarantee that the graph class is preserved.
The reduction is included for completeness and, combined with Eq.~\eqref{eq:bunke}, does not yield complexity results for \GED beyond its \NP-hardness.

\subsubsection{\SI reduces to \GED under unit costs.}
The following equivalence has been established by Zeng et al.~\cite{zengComparingStarsApproximating2009}.
The graph $G_1=(V_1, E_1)$ is isomorphic to a subgraph of $G_2=(V_2,E_2)$ if and only if
$|V_1| \le |V_2|$ and $|E_1| \le |E_2|$ and
\begin{equation}\label{eq:SItoGED}
  \GED(G_1, G_2; C_U) = |V_2| - |V_1| + |E_2| - |E_1|.   
\end{equation}
Since \SI is \NP-complete~\cite{garey1979computers}, \GED is shown to be \NP-hard under unit costs~\cite{zengComparingStarsApproximating2009}.
Furthermore, we can transfer the hardness results known for \SI in restricted graph classes.
It directly follows that \GED under unit costs is \NP-hard even if $G_1$ is a forest and $G_2$ is a tree~\cite[Theorem 4.6]{garey1979computers}. Moreover, it remains \NP-hard in outerplanar graphs, even if $G_1$ is connected~\cite{SYSLO198291}, and in partial $k$-trees unless $G_1$ is $k$-connected or has at most $k$ vertices of unbounded degree~\cite{guptaComplexitySubgraphIsomorphism1996}.
On the other hand, \SI is polynomial-time solvable in trees~\cite{MATULA197891}, and in partial $k$-trees if $G_1$ is $k$-connected or connected and of bounded degree~\cite{MATOUSEK1992343}.

\section{Reducing \MCES to \GED Preserving Graph Classes}
In contrast to \SI, the problem \MCES is known to remain \NP-hard even when both input graphs are trees; see Section~\ref{sec:hardness} for details. We provide a reduction from \MCES to \GED, transferring this result to \GED and refining the complexity status obtained from the reductions in Section~\ref{sec:reductions}. In contrast to \MCIS, reducing \MCES to \GED does not require preserving non-adjacencies, eliminating the need to use complete graphs to track non-adjacent vertices. While the reduction appears straightforward, we are not aware of any explicit reference in the literature.

\begin{theorem}\label{thm:GEDMCES}
  Let $G_1=(V_1,E_1)$ and $G_2=(V_2,E_2)$ be two labeled graphs and $C_E$ as defined in Table~\ref{tab:cost-functions}, then
  \begin{equation}\label{eq:edge}
  \GED(G_1, G_2; C_E) = |E_1| + |E_2| - 2\,\MCES(G_1, G_2).
\end{equation}
\end{theorem}
\begin{proof}
We prove the following stronger statements: (i) for any restricted edit path $\gamma$ from $G_1$ to $G_2$, there is a common subgraph $S=(V_S,E_S)$ of $G_1$ and $G_2$ satisfying $C_E(\gamma)=|E_1|+|E_2|-2|E_S|$, and conversely, (ii) for any common subgraph $S$ of $G_1$ and $G_2$, there is a restricted edit path $\gamma$ from $G_1$ to $G_2$ satisfying the same identity.

Proof of (i): Given a restricted edit path $\gamma$ from $G_1$ to $G_2$, we apply all vertex and edge deletions to $G_1$ and additionally delete all relabeled edges to obtain a graph $G_1'$. Let $\phi$ be the node map induced by $\gamma$, which exists according to~\cite{bougleuxGraphEditDistance2017,BlumenthalG20}. For all insertions of a vertex $v$ in $\gamma$, delete the vertex $\phi(v)$ in $G_2$; for all edge insertions and relabelings of edges $(u,v)$ in $\gamma$, delete the corresponding edge $(\phi(u),\phi(v))$ in $G_2$. We denote the resulting graph by $G'_2$. The node map $\phi$ restricted to the vertices preserved in $G'_1$ is an isomorphism between $G'_1$ and $G'_2$, and we thus have identified a common subgraph. Let $e_d$, $e_i$, and $e_r$ denote the number of edge deletions, insertions and relabelings. Then the number of edges in the common subgraph is $|E_1|-e_d-e_r = |E_2|-e_i-e_r$, and the edit costs are $C_E(\gamma) = e_d+e_i+2e_r$, since relabelings have cost two and vertex insertions and deletions have zero costs. For the right-hand side of Eq.~\eqref{eq:edge}, we obtain $|E_1|+|E_2|-(|E_1|-e_d-e_r + |E_2|-e_i-e_r) = e_d+e_i+2e_r$, proving the statement.

Proof of (ii): We start from the isomorphism $\phi$ between common subgraphs of $G_1$ and $G_2$ and derive an edit path $\gamma$ as follows: All edges of $G_1$ not preserved in $G_2$ under $\phi$ are deleted, additional edges $(u,v)$ in $G_2$ are inserted as $(\phi^{-1}(u),\phi^{-1}(v))$ in $G_1$, preserved edges with non-matching labels are relabeled, leading to the same cost as deletion followed by insertion. All edges in $G_1$ and $G_2$ having at least one endpoint not contained in $\phi$ are deleted and inserted, respectively. The vertices in $G_1$ and $G_2$ that are not part of the common subgraph are deleted and inserted, respectively, at zero cost. 
In total, each edge not contained in the common subgraph $S$ incurs a cost of one, and the identity holds. 

From the two statements just proven, it follows, in particular, that the cost of an optimal edit path under the cost function $C_E$ and the number of edges not contained in an \MCES are equal, and Eq.~\eqref{eq:edge} holds. \qed
\end{proof}

Theorem~\ref{thm:GEDMCES} directly yields a polynomial-time reduction from \MCES to \GED under a special cost function preserving the graph classes.
For the case of unlabeled graphs, \MCES is equivalent to \GSim.
Given two graphs of the same order $n$, it holds
\begin{equation}
\label{eq:GSimToMCES}
  \mathrm{dist}(G_1, G_2)^2 = 2|E_1| + 2|E_2| - 4\,\MCES(G_1, G_2) = 2\,\GED(G_1, G_2; C_E),
\end{equation}
since the adjacency entries are binary, and each mismatch of an undirected edge contributes two differing matrix entries. The result extends to graphs of different sizes by padding the smaller graph with isolated vertices.
From the result that \GSim is \NP-hard if $G_1$ and $G_2$ are trees~\cite{grohe_et_al:LIPIcs.MFCS.2018.20}, we obtain that \MCES and \GED are \NP-hard on trees. 
Please note that there are various polynomial-time computable variants of the tree edit distance, such as~\cite{zhang_simple_1989}, which, however, are not equivalent to \GED on trees.
\GSim can be computed in polynomial-time when $G_1$ is a path and $G_2$ is a tree~\cite{grohe_et_al:LIPIcs.MFCS.2018.20}. We show that this positive result does not transfer to \MCES and \GED in labeled graphs.

\section{Hardness of \MCES on Labeled Paths}\label{sec:hardness}
The complexity of \MCES has been extensively investigated. Garey and Johnson~\cite[GT49]{garey1979computers} show its \NP-hardness and refer to a private communication with Edmonds and Matula for the result that the problem can be solved in polynomial time on two trees. However, as suggested in \cite{MATULA197891}, the communication probably referred to the case where the common subgraph must be connected, a problem known as the maximum common subtree problem~\cite{droschinskyFasterAlgorithmsMaximum2016}.
Indeed, there are papers~\cite{Hajiaghayi2007,kriege_chemical_2025} stating that \MCES remains \NP-hard in trees, or---more generally---that finding a ($k-1$)-connected \MCES in partial $k$-trees is \NP-hard, referring to an unpublished manuscript~\cite{Brandenburg2000}.
More recently, Grohe et al.~\cite[Theorem 8]{grohe_et_al:LIPIcs.MFCS.2018.20} showed formally that \GSim is \NP-hard when both input graphs are trees. Using Eq.~\eqref{eq:GSimToMCES}, we obtain the same hardness result for \MCES, which has recently been shown directly by reduction from $3$-Partition~\cite{Rautenbach_Werner_2026}.
Grohe et al.~\cite[Theorem 9]{grohe_et_al:LIPIcs.MFCS.2018.20} give a polynomial-time algorithm for \GSim in unlabeled graphs, when one graph is a tree and the other a path, which directly extends to \MCES.
We consider the labeled setting and provide a natural extension of the $3$-Partition-based reductions showing that \MCES for two edge-labeled paths is \NP-hard.
By Theorem~\ref{thm:GEDMCES}, this implies \NP-hardness of the \GED on edge-labeled paths.

\begin{figure}[t]   
    \centering

    \begin{tikzpicture}[
        dot/.style={circle, fill=black, inner sep=0pt, minimum size=4pt},
        lab/.style={font=\scriptsize},
        solid/.style={semithick},
        dashed sep/.style={semithick, dashed},
        dashdot sep/.style={semithick, dash dot},
        brace above/.style={decorate, decoration={brace, amplitude=4pt, raise=5pt}},
        brace below/.style={decorate, decoration={brace, amplitude=4pt, raise=5pt, mirror}},
    ]
    
    \def\sp{0.7}

    \foreach \i in {0,...,3} { \node[dot] (a1-\i) at (\i*\sp, 0) {}; }
    \draw[solid] (a1-0) -- (a1-1);
    \draw[dotted, thick] (a1-1) -- (a1-2);
    \draw[solid] (a1-2) -- (a1-3);
    \draw[brace above] (a1-0.north) -- (a1-3.north) node[midway, above=7pt, lab] {$a_1$};
    
    \foreach \i in {0,...,3} { \node[dot] (a2-\i) at ({(4.2+\i)*\sp}, 0) {}; }
    \draw[dashed sep] (a1-3) -- (a2-0) node[midway, above, lab] {$2$};
    \draw[solid] (a2-0) -- (a2-1);
    \draw[dotted, thick] (a2-1) -- (a2-2);
    \draw[solid] (a2-2) -- (a2-3);
    \draw[brace above] (a2-0.north) -- (a2-3.north) node[midway, above=7pt, lab] {$a_2$};
    
    \node[dot] (s2) at ({8.4*\sp}, 0) {};
    \draw[dashed sep] (a2-3) -- (s2) node[midway, above, lab] {$2$};
    
    \node at (9.8*\sp, 0) {\large$\cdots$};
    \node[dot] (s3) at (11.2*\sp, 0) {};
    
    \foreach \i in {0,...,3} { \node[dot] (ae-\i) at ({(12.4+\i)*\sp}, 0) {}; }
    \draw[dashed sep] (s3) -- (ae-0) node[midway, above, lab] {$2$};
    \draw[solid] (ae-0) -- (ae-1);
    \draw[dotted, thick] (ae-1) -- (ae-2);
    \draw[solid] (ae-2) -- (ae-3);
    \draw[brace above] (ae-0.north) -- (ae-3.north) node[midway, above=7pt, lab] {$a_{3m}$};
    
    \node[font=\small] at ({7.7*\sp}, -0.7) {$P_1$};

    \begin{scope}[xshift={0.3*\sp cm}]
    
    \foreach \i in {0,...,5} { \node[dot] (b1-\i) at (\i*\sp, -1.8) {}; }
    \draw[solid] (b1-0) -- (b1-1) -- (b1-2);
    \draw[dotted, thick] (b1-2) -- (b1-3);
    \draw[solid] (b1-3) -- (b1-4) -- (b1-5);
    \draw[brace below] (b1-0.south) -- (b1-5.south) node[midway, below=7pt, lab] {$A+2$};
    
    \node[dot] (bar1) at (6.2*\sp, -1.8) {};
    \draw[dashdot sep] (b1-5) -- (bar1) node[midway, above, lab] {$3$};
    
    \node at (7.4*\sp, -1.8) {\large$\cdots$};
    
    \node[dot] (bar2) at (8.6*\sp, -1.8) {};
    
    \foreach \i in {0,...,5} { \node[dot] (bm-\i) at ({(9.8+\i)*\sp}, -1.8) {}; }
    \draw[dashdot sep] (bar2) -- (bm-0) node[midway, above, lab] {$3$};
    \draw[solid] (bm-0) -- (bm-1) -- (bm-2);
    \draw[dotted, thick] (bm-2) -- (bm-3);
    \draw[solid] (bm-3) -- (bm-4) -- (bm-5);
    \draw[brace below] (bm-0.south) -- (bm-5.south) node[midway, below=7pt, lab] {$A+2$};
    
    \node[font=\small] at (7.4*\sp, -2.8) {$P_2$};
    
    \end{scope}
    
    \end{tikzpicture}
    \caption{Construction of paths $P_1$ and $P_2$ for \NP-hardness reduction.}
    \label{fig:path_construction}
\end{figure}
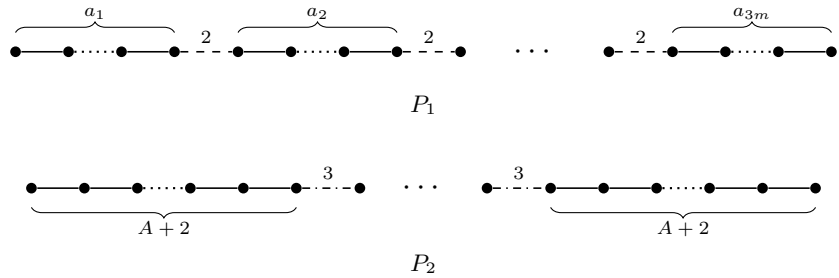

\begin{theorem}

\label{thm:mcis_degree}
\MCES is \NP-hard even when restricted to edge-labeled paths.
\end{theorem}

\begin{proof}
 The proof is by a reduction from $3$-Partition, which is \NP-hard in the strong sense and defined as follows. 
 Given an integer $A$ and a multiset $S = \{a_1, \dots, a_{3m}\}$ with $\sum_i a_i = mA$ and $\nicefrac{A}{4} \leq a_i \leq \nicefrac{A}{2}$, the question is if $S$ can be partitioned into $m$ triples each summing to exactly $A$. Given such an instance, we construct two labeled paths $P_1$ and $P_2$ as follows.

For each of the $3m$ elements $a_i$, construct a path of $a_i$ edges with label $1$.
Connect these $3m$ paths according to their index by edges with label $2$, yielding a single path $P_1$ of size $mA + 3m - 1$.
Analogously, construct $m$ paths of $A + 2$ edges with label $1$ and connect them with edges of label $3$.
This results in a single path $P_2$ of length $mA + 3m - 1$. Set $K = mA$.

A $3$-Partition instance has answer YES if and only if $\MCES(P_1, P_2) \ge K$.
Since label $2$ edges appear only in $P_1$ and label $3$ edges appear only in $P_2$, no such edge can be part of any labeled common subgraph, so only label $1$ edges contribute. In $P_1$, these form $3m$ disjoint paths of sizes $a_1, \dots, a_{3m}$, and in $P_2$, they form $m$ disjoint paths of $A + 2$ edges each. We map the three paths of each triple into the corresponding path of $P_2$, yielding a common subgraph of $mA$ edges.

Conversely, suppose a common subgraph $H$ of size at least $mA$ exists. Analogously, $H$ can only be composed of label $1$ edges. Since $P_1$ contains exactly $mA$ such edges, all must be included in $H$. In particular, each $a_i$ path must be mapped entirely into a single path of $P_2$, since splitting it would lose an edge and violate $|H| \ge mA$. With $3m$ paths distributed over $m$ paths, each of capacity $\nicefrac{A}{4} \leq a_i \leq \nicefrac{A}{2}$ and $A + 2$, each path in $P_2$ receives exactly three $a_i$ paths of $P_1$, since fewer would leave edges of $P_1$
unmapped. Since the total is $mA$, each sums to exactly $A$, yielding a valid $3$-Partition. \qed
\end{proof}

\section{Conclusion}

We presented a unified view of several graph similarity problems and their reductions to \GED.
While the complexity of \GED is typically derived from reductions from \SI, the literature does not distinguish between restricted graph classes and edit cost functions in a fine-grained manner. We have shown that the GED remains \NP-hard even on graph classes admitting polynomial-time \SI algorithms, thereby refining the complexity status. Furthermore, via a straightforward reduction from 3-Partition, we proved that \GED remains \NP-hard even when both input graphs are paths.

As future work, we believe that the joint study of graph classes and cost functions can yield further insights into the complexity landscape of \GED. Moreover, while approximation and parameterized algorithms for \SI are well studied (see, e.g.,~\cite{MarxP14}), only few works address maximum common subgraph problems~\cite{Kann1992,Abu-KhzamBS17}, and even less is known for \GED. Extending these algorithmic paradigms to \GED and investigating the associated complexity is a promising direction for future research.

\subsubsection*{Acknowledgements.}
This work has been supported by the Vienna Science and Technology Fund (WWTF) and the City of Vienna [Grant ID: 10.47379/ ICT22059].

\bibliographystyle{splncs04}
\bibliography{references}

\end{document}